\documentclass[letterpaper, 10pt, conference]{ieeeconf}  

\IEEEoverridecommandlockouts                              
\usepackage{graphics} 
\usepackage{epsfig} 
\usepackage{amsmath} 
\usepackage{amssymb}  
\usepackage{cite}
\usepackage{comment}
\usepackage{xcolor}

\let\labelindent\relax
\usepackage{enumitem}

\usepackage{amsthm}

\renewenvironment{proof}
{\par\noindent\hspace{2em}\textit{Proof:}\ }
{\hspace*{\fill}$\square$\par}

\usepackage{float}

\newtheorem{theorem}{Theorem}

\newtheorem{lemma}{Lemma}

\newtheorem{corollary}{Corollary}

\title{\LARGE \bf Observability Blocking in a Linear Synchronization Network with Partial State Measurements}

\author{Alexis Moreno$^{1}$ and Abdullah Al Maruf$^{2}$
     \thanks{$^1$California State University, Los Angeles (Cal State LA).
        {\tt\small amoren133@calstatela.edu@calstatela.edu}}
     \thanks{$^2$California State University, Los Angeles (Cal State LA).
        {\tt\small amaruf@calstatela.edu}}	
	}

\begin{document}

\maketitle

\begin{abstract}

Large-scale networked systems are increasingly vulnerable to adversaries that can infer system dynamics from a small set of compromised nodes. This paper addresses the problem of blocking such inference using limited state information. While existing state-feedback methods achieve observability blocking with eigenvalue preservation, they require full state access and are impractical for large networks. We propose multiple control strategies that operate under partial state measurements. The first approach employs output feedback to achieve observability blocking while preserving a subset of open-loop eigenvalues. The second approach leverages an observer to reconstruct the system state and enables full-state feedback control, preserving all eigenvalues and providing greater flexibility in sensor placement. We further extend the observer-based design to a distributed framework. Numerical examples demonstrate the scope and validity of the proposed methods.


\end{abstract}


\section{Introduction}

The controllability and observability of dynamical networks have been widely studied in control systems literature \cite{rahmani2009controllability,pasqualetti2014controllability,summers2015submodularity,roy2016sensor,li2020structural}. More recently, attention has shifted toward the design of control strategies that shape these properties while preserving overall network performance  \cite{roy2016sensor,li2020structural,paridari2017framework}. This problem becomes particularly important in settings where multiple agents, including adversaries, have access to the network dynamics and may pursue conflicting objectives. For instance, growing concerns about cyber-attacks in power systems have motivated the development of wide-area control schemes that prevent adversaries from estimating or manipulating system dynamics \cite{sridhar2011cyber}. Similarly, multi-vehicle systems must be designed to remain secure against intruders that probe a subset of agents to infer the global state \cite{xue2014security}. In this context, observability-blocking has emerged as a promising approach for enhancing privacy and security in dynamical networks \cite{zhang2023observability,zhang2022privacy}.

Based on this motivation, our previous work demonstrated that state-feedback controllers can successfully block observability at selected nodes of a network by modifying only a few closed-loop eigenvectors while preserving the remaining eigenstructure \cite{al2022observability,tran2025observability,anguluri2025mode,al2019observability}. However, a key limitation of these approaches is the assumption that the operator has access to the full system state. In large-scale networks, sensing is inherently distributed, such that only a small subset of nodes can be measured, and communication constraints prevent continuous access to state information of the entire network \cite{zhang2017distributed}. These physical and operational constraints make the classical state-feedback-based solution approach impractical and motivate the development of solutions that operate using only partial measurements available to the operator. Our aim in this study is to address this gap. By combining graph-theoretic insights with linear-algebraic techniques, we develop output-feedback and observer-based control schemes that preserve desired dynamical properties while preventing reconstruction of critical states at vulnerable locations by an adversary or unauthorized stakeholder.

The main contributions of the study are summarized as follows:

\begin{enumerate}
    \item We provide an output-feedback control scheme, presented in a general algorithm, that blocks the adversary’s observability at a set of compromised nodes while relying only on partial measurements collected at designated sensor nodes. We discuss the feasibility and scope of this approach.
    
    \item We present an alternative design approach that employs a state observer to estimate the system's states using the measurements from the sensor nodes, together with a state-feedback control scheme that uses the estimated state. We show that this approach provides greater flexibility in the placement of sensor nodes and can preserve all open-loop eigenvalues. 

    \item We also present a distributed framework for observability blocking by extending the observer-based approach. We detail constructions of the observer and controller.

    \item Simulation case studies are provided to illustrate the design procedures and validate the theoretical results.

    \end{enumerate}

This paper is organized as follows. Section II defines the
network model and problem formulation. Section III presents
our main results, and Section IV showcases a simulation study.
Section V concludes the paper.

\section{Network Model and Problem Formulation}

In this study, we consider the standard model of a linear synchronization network consisting of $n$ nodes. The nodes are interconnected
through edges of certain weights that represent how strongly one node influences another. The network can be described as a weighted directed digraph $\mathcal{G}(\mathcal{V},\mathcal{E};\mathcal{W})$, where $\mathcal{V} = \{1, 2,\ldots,n\}$ is the vertex set, $\mathcal{E} \subseteq \mathcal{V}\times\mathcal{V}$ is the set of directed edges, and $\mathcal{W} = \{ w_{ik} > 0 : (i,k)\in\mathcal{E} \}$ contains positive edge weights. We define the neighbor nodes of $i$ as $\mathcal{N}(i)=\{k:(k,i) \in\mathcal{E}\}$. Each node $i$ has a scalar state $x_i (t) \in \mathbb{R}$ that evolves in continuous time. 

Within the network, we distinguish three subsets of nodes: (i) $q$ actuation nodes, indexed by $\{r_1, r_2, \hdots, r_q\}$, where the system operator can apply control inputs;  (ii) $h$ sensor nodes, indexed by $\{z_1, z_2, \hdots, z_h\}$ whose states can be accessed by the operator; and (iii) $m$ compromised nodes, indexed by $\{s_1, s_2, \hdots, s_m\}$, whose states can be accessed by an adversary. The set of sensor nodes may overlap with the sets of actuation nodes and compromised nodes.


For convenience in presenting the synchronization dynamics, we define the Laplacian matrix $\mathbf{L} \in \mathbb{R}^{n \times n}$ associated with the graph $\mathcal{G}(\mathcal{V},\mathcal{E};\mathcal{W})$. For each pair of vertices $(i,k)$, the off–diagonal entry of $\mathbf{L}$ is $L_{ki} = -w_{ik}$ if $(i,k)\in\mathcal{E}$ and $L_{ki}=0$ otherwise. The diagonal entries are given by $L_{kk} = -\displaystyle\sum_{\substack{i =1, i\neq k}}^{n} L_{ki}$. 
The system dynamics can now be described by the following model:
\begin{align}
\dot{\mathbf{x}}(t) &= -\mathbf{L}\mathbf{x}(t) + \mathbf{B}\mathbf{u}(t), \label{eq:dynamics1} 
\end{align}
where $\mathbf{x}(t)= [x_1, x_2, \cdots x_n]^T\in\mathbb{R}^n$ is the network's state vector and $\mathbf{u}(t)\in\mathbb{R}^q$ is a vector of 
inputs applied at the actuation nodes. Here, the state matrix $\mathbf{L}$ is the Laplacian of the graph, representing the standard synchronization dynamics. The input matrix $\mathbf{B} \in \mathbb{R}^{n \times q}$ is given by $\mathbf{B} = \begin{bmatrix} \mathbf{e}_{r_1}, \mathbf{e}_{r_2}, \ldots , \mathbf{e}_{r_q} \end{bmatrix}$. The notation $e_i\in\mathbb{R}^{n}$ denotes the indicator vector whose $i$-th entry is $1$ and all other entries are $0$. Here, the pair $(-\mathbf{L}, \mathbf{B})$ is assumed to be controllable.

The operator's observed states from the sensor nodes are given by:
\begin{align}
\mathbf{y}_O(t) &= \mathbf{C}_O \mathbf{x}(t), \label{eq:dynamics2}
\end{align}
where $\mathbf{y}_O(t) \in\mathbb{R}^h$ is the vector of measurements observed at the sensor nodes. The corresponding output matrix $\mathbf{C}_O \in \mathbb{R}^{h \times n}$ is defined by $\mathbf{C}_O = \begin{bmatrix} \mathbf{e}_{z_1}, \mathbf{e}_{z_2}, \ldots,\mathbf{e}_{z_h}\end{bmatrix}^T$. 

On the other hand, an adversary can observe the system through: 
\begin{align}
\mathbf{y}_A(t) = \mathbf{C}_A \mathbf{x}(t), \label{eq:dynamics3}
\end{align}
where $\mathbf{y}_A(t) \in\mathbb{R}^m$ is the vector of measured states at the compromised nodes. The adversary's output matrix $\mathbf{C}_A \in \mathbb{R}^{m \times n}$ is given by $\mathbf{C}_A = \begin{bmatrix} \mathbf{e}_{s_1}, \mathbf{e}_{s_2}, \ldots,\mathbf{e}_{s_m}\end{bmatrix}^T$.

In this study, we aim to block observability at the compromised nodes, so that the adversary cannot estimate the network state $\mathbf{x}(t)$ from the measurement $\mathbf{y}_A(t)$. To do so, we apply feedback controls at the actuation nodes to prevent observability at the compromised nodes. Unlike our previous work, the proposed approach relies only on partial state information, namely the measurements $\mathbf{y}_O(t)$ obtained from the sensor nodes. Specifically, we will design feedback controls based on the measurements $\mathbf{y}_O(t)$ such that the closed-loop model is unobservable to the measurement $\mathbf{y}_A(t)$. In addition, we seek to preserve key open-loop properties, particularly the open-loop eigenstructure, as much as possible, while also reducing the number of required actuation nodes.

\section {Main Results}
In this section, we present the main results of our study. First, we present our results for the output-feedback-based approach to block observability at the compromised nodes. 
Next, we present the observer-based approach, providing designs for constructing a state observer using measurements from the sensor nodes and a state-feedback controller at the actuation nodes to block observability. Then, we show how the observer-based approach can be leveraged to develop a distributed control framework for blocking observability. 

Before presenting the main result, we make some remarks on the eigenstructure of the Laplacian dynamics. For simplicity, we assume that the eigenvalues of $\mathbf{L}$ are distinct. 
Let 
$\{\lambda_1,\lambda_2,\ldots,\lambda_n\}$ denote the eigenvalues of $\mathbf{L}$ with the
corresponding eigenvectors $\{\mathbf{v}_1,\mathbf{v}_2,\ldots,\mathbf{v}_n\}$. In addition, $\mathbf{V}_0 = [\,\mathbf{v}_1\;\mathbf{v}_2\;\cdots\;\mathbf{v}_n\,]\in\mathbb{C}^{n\times n}$  denotes the open-loop modal matrix of 
$\mathbf{L}$. For our convenience, we also define the set of open-loop eigenvalues and eigenvectors as $\Lambda_0$ and $V_0$, respectively. This spectral decomposition forms the foundation of our algorithms developed in this section, as we leverage eigenstructure assignment techniques to impose unobservability on any selected mode through an appropriately designed feedback control law. 

\subsection{Output-Feedback Based Approach}

In the first approach, we employ output feedback to block observability at the compromised nodes. Specifically, we consider a static output-feedback control scheme in which the input vector $\mathbf{u}(t)$ is specified as:
\begin{equation}
    \mathbf{u}(t) = -\mathbf{F}_1 \mathbf{y}_O(t)=-\mathbf{F}_1 \mathbf{C}_O \mathbf{x}(t),
    \label{eq:input_for_OP_FB_control}
\end{equation}
where $\mathbf{F}_1 \in \mathbb{R}^{q \times h}$ is the output-feedback gain matrix. Hence, for this control law, the closed-loop dynamics are given by:
\begin{equation}
    \dot{\mathbf{x}}(t) =  -(\mathbf{L} + \mathbf{B} \mathbf{F}_1 \mathbf{C}_O)\mathbf{x}(t).
    \label{eq:OP_FB_closed_loop_dynamics}
\end{equation}
Our goal in this approach is to design $\mathbf{F}_1$ such that the closed-loop dynamics $(\mathbf{C}_A,-(\mathbf{L} + \mathbf{B F}_1 \mathbf{C}_O))$ becomes unobservable. We design the gain matrix $\mathbf{F}_1$ by leveraging an output-feedback-based eigenstructure assignment technique \cite{srinathkumar2003eigenvalue, andry2007eigenstructure}. Specifically, we adapt and extend existing output-feedback methods to enforce unobservability at the compromised nodes while preserving as much of the open-loop eigenstructure as possible. To do so, we first select a self-conjugate subset of $h$ open-loop eigenvalues of $\mathbf{L}$ that we want to preserve in the closed-loop dynamics $-(\mathbf{L+BF}_1\mathbf{C}_O)$, denoted as
$\Lambda_D= \{\lambda_{i_1},\lambda_{i_2},\ldots,\lambda_{i_h}\}$, 
where $\{i_1,i_2,\ldots,i_h\} \subseteq \{1,2,\ldots,n\}$. We will place these eigenvalues and their associated eigenvectors through the output-feedback control scheme. For each selected  eigenvalue $\lambda_{i_k}\in \Lambda_D$, the corresponding open–loop eigenvector 
is the column $\mathbf{v}_{i_k}$ of $\mathbf{V}_0$.

Our proposed observability blocking design algorithm allows us to choose any eigenvalue $\lambda_{i_p}$ from the set $\Lambda_D$ together with its associated eigenvector $\mathbf{v}_{i_p}\in \mathbf{V}_0$, where $i_p \in \{i_1, i_2,\ldots, i_h\}$. 
We will modify $\mathbf{v}_{i_p}$ to a new vector $\mathbf{\hat{v}}_{i_p}$, whose entries corresponding to the compromised nodes become zero. Thus, the mode $(-\lambda_{i_p})$ becomes unobservable according to the well-known Popov–Belevitch–Hautus (PBH) test \cite{rugh1996linear} since $\mathbf{C}_A \hat{\mathbf{v}}_{i_p} = 0$. Meanwhile, the remaining eigenvalues in $\Lambda_D$ and their corresponding eigenvectors are preserved in the closed-loop system.

We present our observability blocking algorithm below. For convenience in presenting the algorithm, we assume that the actuation nodes are labeled as the first $q$ nodes (i.e., nodes $1, 2, \cdots, q$), while the compromised nodes are the last $m$ nodes (i.e., nodes $n-m+1, \cdots, n-1, n$) in the following discussion. This assumption does not impose any restriction, as the nodes can always be relabeled without loss of generality. Our proposed observability blocking control design algorithm involves the following steps:

\vspace{10pt}
\noindent\textbf{Algorithm 1.}\\

\begin{enumerate}[label={(\arabic*)}]

\item Select the set $\Lambda_D= \{\lambda_{i_1},\lambda_{i_2},\cdots,\lambda_{i_h}\}$ from the eigenvalues of $\mathbf{L}$ and then choose one eigenvalue $\lambda_{i_p}$ from $\Lambda_D$ and its associated eigenvector 
$\mathbf{v}_{i_p}$, where $i_p \in \{i_1, i_2,\ldots, i_h\}$. 

\vspace{0.3cm}

\item Compute a matrix $\mathbf{N}(\lambda_{i_p}) \in \mathbb{C}^{(n+q)\times q}$, whose columns are 
linearly independent and span the null space of $\mathbf{S}(\lambda_{i_p}) = \big[ (\mathbf{L} - \lambda_{i_p} \mathbf{I}_n)\; \mathbf{B} \big]$. Then, partition $\mathbf{N}(\lambda_{i_p})$ as $\mathbf{N}(\lambda_{i_p}) =
\begin{bmatrix}
\mathbf{N}_1(\lambda_{i_p})^T\;\mathbf{N}_2(\lambda_{i_p})^T
\end{bmatrix}^T$, where $\mathbf{N}_1(\lambda_{i_p}) \in \mathbb{C}^{n \times q}$ and $\mathbf{N}_2(\lambda_{i_p}) \in \mathbb{C}^{q \times q}$.
Therefore, $\mathbf{N}_1(\lambda_{_p})$ and $\mathbf{N}_2(\lambda_{_p})$ satisfy the following where $\mathbf{I}_n$ denotes an identity matrix in $\mathbb{R}^{n \times n}$:
\begin{align}
\big[ (\mathbf{L} - \lambda_p \mathbf{I}_n)\; \mathbf{B} \big]
\begin{bmatrix}
\mathbf{N}_1(\lambda_{i_p}) \\[1pt]
\mathbf{N}_2(\lambda_{i_p})
\end{bmatrix}
= \mathbf{0}.
\tag{4} \label{eq_prf1_2}
\end{align}

\item Partition $\mathbf{N}_1(\lambda_{i_p})$ as
\[
\mathbf{N}_1(\lambda_{i_p})
=
\begin{bmatrix}
\mathbf{N}_3(\lambda_{i_p})^T \\[1pt]
\mathbf{N}_4(\lambda_{i_p})^T
\end{bmatrix}^T,
\]
where $\mathbf{N}_3(\lambda_{i_p}) \in \mathbb{C}^{(n-m)\times q}$ and
$\mathbf{N}_4(\lambda_{i_p}) \in \mathbb{C}^{m\times q}$. 
Then, find a vector $\mathbf{h}_{i_p} \neq \mathbf{0}$ which lies in the null space 
of $\mathbf{N}_4(\lambda_{i_p})$, i.e.,
\[
\mathbf{N}_4(\lambda_{i_p})\,\mathbf{h}_{i_p} = \mathbf{0}.
\tag{5} \label{eq:th1:prf7}
\]

\item Compute the vector $\mathbf{\hat{v}}_{i_p}$ as:
\[
\mathbf{\hat{v}}_{i_p} = \mathbf{N}_1(\lambda_{i_p})\, \mathbf{h}_{i_p}.
\tag{6} \label{eq:th1_3}
\]

\item Form the $n \times h$ zero matrix $\mathbf{V}_h$ to store the desired $h$
closed-loop eigenvectors. For each eigenvalue in $\Lambda_D$, place
its corresponding open-loop eigenvector from $\mathbf{V}_0$ in the appropriate
column of $\mathbf{V}_h$. For the selected eigenvalue $\lambda_{i_p}$, replace $\mathbf{v}_{i_p}$ with the modified
eigenvector $\hat{\mathbf{v}}_{i_p}$ in $\mathbf{V}_h$. After all columns have been assigned, the
matrix $\mathbf{V}_h$ takes the form
\[
\mathbf{V}_h
=
\big[\,\mathbf{v}_{i_1} \;\ldots\; \mathbf{v}_{i_{p-1}} \;\hat{\mathbf{v}}_{i_p}\;
\mathbf{v}_{i_{p+1}} \;\ldots\; \mathbf{v}_{i_h}\,\big],
\]
If $\lambda_{i_p}$ is complex (i.e., not real), then it must have its complex conjugate $\bar{\lambda}_{i_p}$ also included in $\Lambda_D$. We then also need to replace $\bar{\mathbf{v}}_{i_p}$ with $\bar{\hat{\mathbf{v}}}_{i_p}$ in $\mathbf{V}_h$ where $\bar{\mathbf{v}}_{i_p}$ is the eigenvector corresponding to $\bar{\lambda}_{i_p}$.

\item 
For each eigenvalue $\lambda_{i_k} \in \Lambda_D$, extract the first $q$ entries of the
corresponding eigenvector in $\mathbf{V}_h$. We denote this sub-vector as $\mathbf{z}_{i_k}$. For instance, $\hat{\mathbf{v}}_{ip}= [\mathbf{z}_{i_p}^T~\mathbf{w}_{i_p}^T]^T$ where  $\mathbf{z}_{i_p}\in \mathbb{C}^{q}$ and $\mathbf{w}_{i_p}\in \mathbb{C}^{n-q}$. Then we compute
$\lambda_{i_k} \mathbf{z}_{i_k}$ and place these scaled vectors into the following manner to construct the $q \times h$ matrix $\mathbf{Z}$:
\[
\mathbf{Z}
=
\big[\,\lambda_1 \mathbf{z}_1 \;\; \lambda_2 \mathbf{z}_2 \;\; \cdots \;\;
\lambda_h \mathbf{z}_h\,\big].
\]

\item Partition $\mathbf{L}$ into:
\[
\mathbf{L} =
\begin{bmatrix}
\mathbf{L}_1 \\[4pt]
\mathbf{L}_2
\end{bmatrix},
\]
where $\mathbf{L}_1 \in \mathbb{R}^{q \times n}$ 
and $\mathbf{L}_2 \in \mathbb{R}^{(n-q) \times n}$.

\item Finally, the output-feedback gain matrix $\mathbf{F}_1$ can be obtained by:
\begin{align}
\mathbf{F}_1 = (\mathbf{Z} - \mathbf{L}_1 \mathbf{V}_h)(\mathbf{C}_O \mathbf{V}_h)^{-1}. \label{eq:F1}
\end{align}

\end{enumerate}

The following theorem formalizes that the controller designed according to Algorithm 1 makes the pair $(\mathbf{C}_A,  -(\mathbf{L} + \mathbf{B} \mathbf{F}_1 \mathbf{C}_O)$ unobservable. Furthermore, it preserves a desired set of $h$ eigenvalues specified by $\Lambda_D$, whose associated eigenvectors also remain unchanged except the eigenvector $\mathbf{v}_p$ (and also $\bar{\mathbf{v}}_p$ if $\lambda_{i_p}$ is complex).

\vspace{10pt}
\begin{theorem}
Consider the network synchronization model \eqref{eq:dynamics1}. Assume that: $1)$ the eigenvalues of $\mathbf{L}$ are distinct; 2) $q \geq m + 2$; 3) the matrix $\mathbf{C}_O \mathbf{V}_h$ is invertible; and 4) the pair $(-\mathbf{L},\mathbf{B})$ is controllable. Then, the gain matrix $\mathbf{F}_1$ of the output-feedback controller obtained via Algorithm 1  makes the pair $(\mathbf{C}_A,  -(\mathbf{L} + \mathbf{B} \mathbf{F}_1 \mathbf{C}_O)$  unobservable. Particularly, any selected open-loop mode (-$\lambda_{i_p}$) becomes unobservable to the compromised nodes. Furthermore, all eigenvalues in the set $\Lambda_D$ are preserved, whose associated closed-loop eigenvectors are given by the columns in $\mathbf{V}_h$.
\end{theorem}

\begin{proof}
From \eqref{eq:F1}, we can write 
\begin{align}
(\mathbf{L}_1+\mathbf{F}_1\mathbf{C}_O)\mathbf{V}_h=\mathbf{Z}.    
\end{align}
From the definitions of $\mathbf{V}_h$ and $\mathbf{Z}$, we get
\begin{align}
(\mathbf{L}_1+\mathbf{F}_1\mathbf{C}_O)\hat{\mathbf{v}}_{i_p} &=\lambda_{i_p} \mathbf{z}_{i_p} \label{eq:proof_th1}\\
(\mathbf{L}_1+\mathbf{F}_1\mathbf{C}_O) \bar{\hat{\mathbf{v}}}_{i_p} &=\bar{\lambda}_{i_p}\bar{\mathbf{z}}_{i_p} ~~~\mathrm{when}~\lambda_{i_p}~\mathrm{is~complex} \nonumber\\
(\mathbf{L}_1+\mathbf{F}_1\mathbf{C}_O)\hat{\mathbf{v}}_{i_k} &=\lambda_{i_k} \mathbf{z}_{i_k} ~~~\mathrm{for~all~other}~i_k \nonumber
\end{align}
Now, we focus on \eqref{eq:proof_th1} to show that $(-\lambda_{i_p})$ and $\hat{\mathbf{v}}_{i_p}$ are closed-loop eigenvalue and eigenvector pair. To do so, consider the following partition of $\mathbf{L}$.
\begin{equation*}
\mathbf{L} =
\begin{bmatrix}
\mathbf{L}_1 \\
\mathbf{L}_2
\end{bmatrix}
=
\begin{bmatrix}
\mathbf{L}_{11} & \mathbf{L}_{12} \\
\mathbf{L}_{21} & \mathbf{L}_{22} \nonumber
\end{bmatrix}  
\end{equation*}
where $\mathbf{L}_{1} \in \mathbb{R}^{q \times n}$, $\mathbf{L}_{11} \in \mathbb{R}^{q \times q}$ and the size of the other sub-matrices follows accordingly. Now, from \eqref{eq:proof_th1}, we can write
\begin{align}
\lambda_{i_p}~\mathbf{I}_q~\mathbf{z}_{i_p}-\begin{bmatrix}
\mathbf{L}_{11} &
\mathbf{L}_{12}
\end{bmatrix}~\begin{bmatrix}
\mathbf{z}_{i_p} \\
\mathbf{w}_{i_p}
\end{bmatrix} = \mathbf{F}_1\mathbf{C}_O~\begin{bmatrix}
\mathbf{z}_{i_p} \\
\mathbf{w}_{i_p}
\end{bmatrix} \nonumber
\end{align}
or
\begin{align}
\begin{bmatrix}
\lambda_{i_p}\mathbf{I}_q-\mathbf{L}_{11} &
-\mathbf{L}_{12}
\end{bmatrix}~\begin{bmatrix}
\mathbf{z}_{i_p} \\
\mathbf{w}_{i_p}
\end{bmatrix} = \mathbf{F}_1\mathbf{C}_O~\begin{bmatrix}
\mathbf{z}_{i_p} \\
\mathbf{w}_{i_p}
\end{bmatrix} \label{eq:th1_prf3}
\end{align}
Now, multiplying \eqref{eq_prf1_2} with $\mathbf{h}_{i_p}$ and using \eqref{eq:th1_3}, we have
\begin{align}
\big[ (\mathbf{L} - \lambda_{i_p} \mathbf{I}_n)\; \mathbf{B} \big]
\begin{bmatrix}
\hat{\mathbf{v}}_{i_p} \\
\mathbf{N}_2(\lambda_{i_p})~\mathbf{h}_{i_p}
\end{bmatrix} 
= \mathbf{0}. \label{eq:th1_prf4}  
\end{align}
Since we have labeled the first $q$ nodes as actuation nodes, without loss of generality, we can write $\mathbf{B}=[\mathbf{I}_q^T~~\mathbf{0}^T]^T$. Then, \eqref{eq:th1_prf4} can be written as 
\begin{align*}
\begin{bmatrix}
\lambda_{i_p}\mathbf{I}_q-\mathbf{L}_{11} &
-\mathbf{L}_{12} & \mathbf{I}_q \\
-\mathbf{L}_{21} &
\lambda_{i_p}\mathbf{I}_{n-q}-\mathbf{L}_{22} & \mathbf{0}
\end{bmatrix}~\begin{bmatrix}
\mathbf{z}_{i_p} \\ \mathbf{w}_{i_p}\\
\mathbf{N}_2(\lambda_{i_p})~\mathbf{h}_{i_p}
\end{bmatrix} =
\mathbf{0}
\end{align*}
From the last block row above, we can write
\begin{align}
\begin{bmatrix}
-\mathbf{L}_{21} &
\lambda_{i_p}\mathbf{I}_{n-q}-\mathbf{L}_{22}
\end{bmatrix}~\begin{bmatrix}
\mathbf{z}_{i_p} \\ \mathbf{w}_{i_p}
\end{bmatrix} =
\mathbf{0} \label{eq:th1_prf5}
\end{align}
Combining \eqref{eq:th1_prf3} and \eqref{eq:th1_prf5}, we get
\begin{align}
\begin{bmatrix}
\lambda_{i_p}\mathbf{I}_q-\mathbf{L}_{11} &
-\mathbf{L}_{12} \\
-\mathbf{L}_{21} &
\lambda_{i_p}\mathbf{I}_{n-q}-\mathbf{L}_{22}
\end{bmatrix}~\begin{bmatrix}
\mathbf{z}_{i_p} \\ \mathbf{w}_{i_p}
\end{bmatrix} \nonumber \\~~~= ~\begin{bmatrix}
\mathbf{I}_q \\ \mathbf{0}
\end{bmatrix}
\mathbf{F}_1\mathbf{C}_O~\begin{bmatrix}
\mathbf{z}_{i_p} \\
\mathbf{w}_{i_p}
\end{bmatrix} \label{eq:th1_prf6}
\end{align}
From \eqref{eq:th1_prf6}, we have $(\lambda_{i_p} \mathbf{I}_n-\mathbf{L})\hat{\mathbf{v}}_{i_p}=\mathbf{B}\mathbf{F}_1\mathbf{C}_O \hat{\mathbf{v}}_{i_p}$ or $(\mathbf{L+BF}_1\mathbf{C}_O) \hat{\mathbf{v}}_{i_p}=\lambda_{i_p}\hat{\mathbf{v}}_{i_p}$ or $-(\mathbf{L+BF}_1\mathbf{C}_O) \hat{\mathbf{v}}_{i_p}=(-\lambda_{i_p})\hat{\mathbf{v}}_{i_p}$. Therefore, $(-\lambda_{i_p})$ and $\hat{\mathbf{v}}_{i_p}$ are closed-loop eigenvalue and eigenvector pairs. Since all eigenvalues $\lambda_{i_k}$ selected in $\Lambda_D$ and their associated eigenvectors $\mathbf{v}_{i_k}$ also satisfy the equation
\begin{align*}
\big[ (\mathbf{L} - \lambda_{i_k{}} \mathbf{I}_n)\; \mathbf{B} \big]
\begin{bmatrix}
\mathbf{v}_{i_k} \\
\mathbf{0}
\end{bmatrix}
= \mathbf{0},
\end{align*}
following the same procedure, we can similarly show that all eigenvalues in the set $\Lambda_D$ and their associated eigenvectors in $\mathbf{V}_h$ are preserved in $(\mathbf{L+BF}_1\mathbf{C}_O)$. Now, we need to show that $\mathbf{C}_A\hat{\mathbf{v}}_{i_p}=\mathbf{0}$ completes the proof which directly follows the argument in \cite{al2022observability}. Since $q\geq m+2$, linearly independent $\hat{\mathbf{v}}_{i_p}$ and $\bar{\hat{\mathbf{v}}}_{i_p}$ always exist. Note that the last $m$ entries in $\hat{\mathbf{v}}_{i_p}$ are zero according to \eqref{eq:th1:prf7}. Since all the compromised nodes are labeled as the last $m$ nodes, therefore  $\mathbf{C}_A\hat{\mathbf{v}}_{i_p}=\mathbf{0}$.

\end{proof}

The number of required actuation nodes ($q \geq m+2$) can be reduced in Theorem 1 by exploiting the network's graph structure. Specifically, when a vertex cutset exists that separates the actuation nodes from the compromised nodes, we can directly invoke the results on sparser observability blocking from \cite{al2022observability} (specifically, see Lemma 1 and Theorem 3 in \cite{al2022observability}), provided that all assumptions of Theorem 1 are met. Although Lemma 1 and Theorem 3 in \cite{al2022observability} were derived for state-feedback, the results and arguments in the proofs are valid for any control law applied to the actuation nodes separated by a vertex cutset from the compromised nodes. Thus, we simply rephrase the result in our context and omit the proof to avoid repetition.

\begin{corollary}
Consider the network synchronization model \eqref{eq:dynamics1}. Suppose there is a cutset $\mathcal{V}_{cut}$ that separates the graph into two partitions $\mathcal{V}_1$ and $\mathcal{V}_2$ in such a way that $\mathcal{V}_1$ does not include any compromised nodes and $\mathcal{V}_2$ does not include any actuation nodes. If the selected $\lambda_{i_p}$ is not an eigenvalue of the block $\mathbf{L}_{\mathcal{V}_2\mathcal{V}_2}$ from the Laplacian matrix $\mathbf{L}$, then the number of actuation nodes required to design the observability-blocking controller in Theorem 1 can be reduced to $q\geq |\mathcal{V}_{cut}|+2$ nodes by designing a controller that makes the mode $(-\lambda_{i_p})$ unobservable for measurements taken at the nodes in $\mathcal{V}_{cut}$.
\end{corollary}

Here, we make some remarks on Condition 3 in Theorem 1. This condition requires the matrix $\mathbf{C}_O \mathbf{V}_h$ to be invertible or equivalently $\operatorname{rank}(\mathbf{C}_O \mathbf{V}_h) = h$. This condition is satisfied when the closed-loop system remains observable at the sensor nodes for all the eigenvalues specified by the set $\Lambda_d$. Therefore, we can say that if the closed-loop model $(\mathbf{C}_O,  -(\mathbf{L} + \mathbf{B} \mathbf{F}_1 \mathbf{C}_O)$ is observable, then $\mathbf{C}_O \mathbf{V}_h$ can be invertible. However, this can be violated if all sensor nodes are located in  $\mathcal{V}_{cut} \cup \mathcal{V}_2$. In that case, Algorithm 1 will fail, as there is no output-feedback control that can impose unobservability on the compromised nodes. We formalize this result below.


\begin{corollary}
Consider the network synchronization model \eqref{eq:dynamics1}. Suppose there exists a vertex cutset $\mathcal{V}_{cut}$ that partitions the graph into two subsets $\mathcal{V}_1$ and $\mathcal{V}_2$ such that $\mathcal{V}_1$ contains no compromised nodes and $\mathcal{V}_2$ contains no actuation nodes. Furthermore, assume that all sensor nodes are located in $\mathcal{V}_{cut} \cup \mathcal{V}_2$, and that all nodes in $\mathcal{V}_{cut}$ are compromised nodes. If $(-\lambda_p)$ is observable for the open-loop model $(\mathbf{C}_O, -\mathbf{L})$ and is not an eigenvalue of the submatrix $\mathbf{L}_{\mathcal{V}_2 \mathcal{V}_2}$, then no output-feedback controller can be designed to render the mode $(-\lambda_p)$ unobservable in the closed-loop model $(\mathbf{C}_A, -(\mathbf{L} + \mathbf{B} \mathbf{F}_1 \mathbf{C}_O))$.
\end{corollary}

\begin{proof}
Observability at the output is invariant under output feedback \cite{rugh1996linear}. Since $(-\lambda_p)$ is observable for the open-loop model $(\mathbf{C}_O, -\mathbf{L})$, it must remain observable in the closed-loop model $(\mathbf{C}_O, -(\mathbf{L} + \mathbf{B} \mathbf{F}_1 \mathbf{C}_O))$. However, since all sensor nodes are located in $\mathcal{V}_{cut} \cup \mathcal{V}_2$ and all nodes in $\mathcal{V}_{cut}$ are compromised, blocking observability at the compromised nodes would also induce unobservability at the sensor nodes by Corollary 1. This contradicts the invariance of output-observability under output feedback. Hence, no output-feedback controller can be designed to make the mode $(-\lambda_p)$ unobservable at the compromised nodes. 
\end{proof}

Several additional notes about the above algorithm are worth mentioning:

\begin{enumerate}

\item If the eigenvalues of $\mathbf{L}$ are real\footnote{For an undirected graph, $\mathbf{L}$ is symmetric and has real eigenvalues only.}, then we require $q = m + 1$ actuation nodes to design the observability-blocking controller. Accordingly, the number of actuation nodes specified in Corollary~1 also decreases to $q= |\mathcal{V}_{cut}|+1$.

\item The output-feedback–based approach preserves only $h$ eigenvalues of the open-loop system. This implies that the remaining $n-h$ eigenvalues will be modified in the closed-loop dynamics. Consequently, having more sensor nodes is advantageous, as it allows more of the open-loop eigenstructure to be retained. However, the modified eigenvalues can potentially destabilize the closed-loop system. Hence, output-feedback-based controls must be applied with caution. 

\item The presented design method is general, in that it does not depend on the state matrix being Laplacian, nor on the specific graph topology. Thereby, this design holds for any linear-time invariant dynamics as long as the conditions given in Theorem 1 are satisfied.


\end{enumerate}

\subsection{Observer Based Approach}

In this subsection, we develop observability-blocking controls that rely on a state observer and full-state feedback control. The state observer reconstructs the full network state $\mathbf{x}(t)$ from the measurements $\mathbf{y}_O(t)$ and the state-feedback controller uses the estimated states for actuation. Specifically, a Luenberger observer is employed to estimate the network's state as given by the following form \cite{siljak2011decentralized, kailath2000linear}:
\begin{align}
\dot{\hat{\mathbf{x}}}(t)
    &= -\mathbf{L}\hat{\mathbf{x}}(t) + \mathbf{B}\mathbf{u}(t)
       + \mathbf{K}\big( \mathbf{y}_O(t) - \mathbf{C}_O\hat{\mathbf{x}}(t) \big), \label{eq:observer} 
\end{align}
where $\hat{\mathbf{x}}(t) \in \mathbb{R}^n$ denotes the estimate of the network's state by the observer,  and $\mathbf{K} \in \mathbb{R}^{n \times h}$ is the observer gain matrix. Then, a state-feedback controller is used based on the estimated state $\hat{\mathbf{x}}(t)$ according to:
\begin{equation}
    \mathbf{u}(t) = -\mathbf{F}_2 \hat{\mathbf{x}}(t),
    \label{eq:observer_system_input}
\end{equation}
where $\mathbf{F}_2 \in \mathbb{R}^{q \times n}$ is the state–feedback controller gain matrix. Our goal in this approach is to design the observer gain $\mathbf{K}$ and controller gain $\mathbf{F}_2$ such that the estimation error $\mathbf{e}(t)= \mathbf{x}(t)-\hat{\mathbf{x}}(t)$ becomes zero quickly and the closed-loop system becomes unobservable at the compromised nodes, i.e., the pair $(\mathbf{C}_A,-(\mathbf{L}+\mathbf{BF}_2)$ becomes unobservable.

The design in this approach begins with the construction of a Luenberger observer given by \eqref{eq:observer}. If the open-loop model $(\mathbf{C}_O, -\mathbf{L})$ is observable, the observer gain $\mathbf{K}$ can be designed using standard pole-placement techniques \cite{chen1984linear} to place the eigenvalues of $-(\mathbf{L} + \mathbf{K}\mathbf{C}_O)$ sufficiently far to the left of the imaginary axis, ensuring that the estimation error $\mathbf{e}(t)$ converges to zero rapidly.  The convergence rate is governed by the eigenvalue of $(\mathbf{L} + \mathbf{K}\mathbf{C}_O)$ with the smallest real part, which is typically chosen to have a large magnitude.
Then, the gain matrix $\mathbf{F}_2$ for state-feedback control is designed using the algorithm given in \cite{al2022observability}. To do so, we select any of the eigenvalues $\lambda_p$ of $\mathbf{L}$ together with its associated eigenvector $\mathbf{v}_p \in \mathbf{V}_0$ where $p \in \{1,2,\ldots,n\}$. The designed gain matrix $\mathbf{F}_2$ modifies $\mathbf{v}_p$ to a new vector $\mathbf{\hat{v}}_p$ whose entries at the compromised nodes are zero to ensure that $\mathbf{C}_A \hat{\mathbf{v}}_p$ = 0. Hence, gain $\mathbf{F}_2$ makes the closed-loop model $(\mathbf{C}_A, -(\mathbf{L} + \mathbf{B F}_2))$ unobservable for the mode $(-\lambda_p)$. Furthermore, it maintains all the open-loop eigenvalues and most of the open-loop eigenvectors. We formalize the outcome of this approach as follows.

\vspace{10pt}
\begin{theorem}
Consider the network synchronization model \eqref{eq:dynamics1}. Assume that: 1) the eigenvalues of $\mathbf{L}$ are distinct; 2) $q \geq m + 2$; 3) the pair $(\mathbf{C}_O,-\mathbf{L})$ is observable; and 4) the pair $(-\mathbf{L}, \mathbf{B})$ is controllable. Then, a state observer with the gain matrix $\mathbf{K}$ can be designed to ensure that the estimation error $\mathbf{e}(t)$ converges to zero rapidly, while a state-feedback controller with the gain matrix $\mathbf{F}_2$ can be designed independently to block the observability of the system at the compromised nodes. Specifically, the pair $(\mathbf{C}_A,\; -(\mathbf{L} + \mathbf{B}\mathbf{F}_2))$ has an unobservable mode at $(-\lambda_p)$ where $\lambda_p$ is any selected eigenvalue of $\mathbf{L}$. Furthermore, all the open-loop eigenvalues will be preserved. 
\end{theorem}

\begin{proof}
This proof primarily relies on the argument of the separation principle \cite{siljak2011decentralized,kailath2000linear} and is presented briefly. Since both the plant and observer have the same input $\mathbf{u}(t) = -\mathbf{F}_2 \hat{\mathbf{x}}(t)$, with some algebraic manipulations, we can show that the augmented dynamics can be written as
\begin{align}
\begin{bmatrix}
\dot{\mathbf{x}} \\
\dot{\mathbf{e}}
\end{bmatrix}
=
\begin{bmatrix}
-(\mathbf{L+BF}_2) & \mathbf{BF}_2 \\
\mathbf{0} & -(\mathbf{L+KC}_O)
\end{bmatrix}
\begin{bmatrix}
\mathbf{x} \\
\mathbf{e}
\end{bmatrix}. 
\end{align}
Because of the block triangular structure and the fact that the pair $(\mathbf{C}_O,-\mathbf{L})$ is observable, $\mathbf{K}$ can be designed independently to place the eigenvalues of $-(\mathbf{L+KC}_O)$ sufficiently far to the left of the imaginary axis, ensuring that the estimation error $\mathbf{e}(t)$ converges to zero rapidly. Therefore, the closed-loop dynamics become $-(\mathbf{L+BF}_2)$. Since all the conditions from Theorem 1 of \cite{al2022observability} are satisfied, we can independently design $\mathbf{F}_2$ such that the pair $(\mathbf{C}_A,\; -(\mathbf{L} + \mathbf{B}\mathbf{F}_2))$ has an unobservable mode at $(-\lambda_p)$ while preserving all the open-loop eigenvalues. 
\end{proof}

We refer the reader to \cite{siljak2011decentralized,kailath2000linear} for the design of the Luenberger observer and to \cite{al2022observability} for the design of an observability-blocking controller using full state feedback. Specifically, we need to use Algorithm 1 from \cite{al2022observability} to compute $\mathbf{F}_2$. We omit presenting the algorithm here again to avoid repetition. 


Similar to the output-feedback-based approach, the number of required actuation nodes can be reduced in Theorem 2 when there exists a vertex cutset that separates the actuation nodes from the compromised nodes. The following corollary formalizes the result.

\begin{corollary}
Consider the network synchronization model \eqref{eq:dynamics1}. Suppose that there is a cutset $\mathcal{V}_{cut}$ which separates the graph into two partitions $\mathcal{V}_1$ and $\mathcal{V}_2$ in such a way that $\mathcal{V}_1$ does not include any compromised nodes and $\mathcal{V}_2$ does not include any actuation nodes. If the selected $\lambda_{p}$ is not an eigenvalue of the block $\mathbf{L}_{\mathcal{V}_2\mathcal{V}_2}$ from the Laplacian matrix $\mathbf{L}$, then the number of actuation nodes required to design the observability-blocking controller in Theorem 2 can be reduced to $q \geq |\mathcal{V}_{cut}|+2$ nodes by designing a controller that makes the mode $(-\lambda_{p})$ unobservable for measurements taken at the nodes in $\mathcal{V}_{cut}$.
\end{corollary}



We highlight several advantages of the observer-based approach over the output-feedback-based approach. First, unlike output feedback, the observer-based approach preserves all eigenvalues of the open-loop system. One might expect that the closed-loop dynamics become unobservable only after the estimation error $\mathbf{e}(t)$ converges to zero. However, in practice, the estimation error effectively acts as an unknown external input to the closed-loop dynamics. As a result, the system remains unobservable to the adversary at all times. 

Second, by comparing Condition 3 in Theorems 1 and 2, we observe that the observer-based design does not require the closed-loop model $(\mathbf{C}_O, -(\mathbf{L} + \mathbf{B}\mathbf{F}_2))$ to remain observable, unlike the output-feedback-based design. In particular, when the closed-loop system becomes unobservable at the sensor nodes while blocking observability at the compromised nodes, the output-feedback-based method fails, as stated in Corollary 2. In contrast, the observer-based approach remains applicable, provided that the open-loop model $(\mathbf{C}_O, -\mathbf{L})$ is observable. To illustrate this, consider the extreme case where the sensor nodes and compromised nodes coincide, i.e., $\mathbf{C}_A = \mathbf{C}_O = \mathbf{C}$. The observer-based method remains applicable as long as all conditions in Theorem 2 are satisfied. This is because the estimation error dynamics for an observer constructed by the operator are governed by the state matrix $-(\mathbf{L} + \mathbf{K}\mathbf{C})$, whereas the estimation error dynamics for an observer constructed by the adversary are governed by the state matrix $-(\mathbf{L} + \mathbf{B}\mathbf{F}_2 + \mathbf{K}\mathbf{C})$. This difference arises because the plant uses the control input $\mathbf{u} = -\mathbf{F}_2 \hat{\mathbf{x}}$, where $\hat{\mathbf{x}}$ is the operator’s state estimate. Consequently, the observer-based approach enables greater flexibility in the placement of sensor nodes, which solely relies on the open-loop model. 

\noindent{Several other remarks on this approach have been made below:}

\begin{enumerate}


\item Similar to the output-feedback-based approach, if the eigenvalues of $\mathbf{L}$ are all real, we require $q = m + 1$ actuation nodes to design the observability-blocking controller. Accordingly, the number of actuation nodes specified in Corollary~3 also reduces to $q= |\mathcal{V}_{cut}|+1$.


\item The presented design method is general, in that it does not depend on the state matrix being Laplacian, nor on the specific graph topology. Therefore, this design holds for any linear-time invariant dynamics as long as the conditions given in Theorem 2 are satisfied.


\item Another important advantage of the observer-based approach is that it provides a framework for deriving distributed solutions to the observability-blocking control problem. We briefly discuss this in the next section.

\end{enumerate}

\subsection{A Distributed Framework}

In this section, we extend the observer-based approach to develop a distributed control framework for observability blocking. In this framework, each actuation node has access only to its own state and those of its neighbors, and uses this information to design its control input $\mathbf{u}_i(t)$. Formally, we define the observation model of each actuation node $i$ as
\begin{align}
\mathbf{y}_{Oi}(t) &= \mathbf{C}_{Oi}~\mathbf{x}(t).\label{eq:dynamics4}
\end{align}
$\mathbf{C}_{Oi}$ is the output matrix for the node $i$ defined as $\mathbf{C}_{Oi} = \begin{bmatrix} \mathbf{e}_{i_1}, \mathbf{e}_{i_2}, \ldots,\mathbf{e}_{i_{|\mathcal{N}(i)|+1}}\end{bmatrix}^T$ where $\{i_1, i_2, \cdots, i_{|\mathcal{N}(i)|+1}\}= \{i\} \cup \mathcal{N}(i)$. In this distributed approach, each actuation node $i$ employs an observer given by the equation
\begin{align}
\dot{\hat{\mathbf{x}}}_i(t)
    &= -\mathbf{L}\hat{\mathbf{x}}_i(t) + \mathbf{B}\mathbf{u}(t)
       + \mathbf{K}_i\big( \mathbf{y}_{Oi}(t) - \mathbf{C}_{Oi}\hat{\mathbf{x}}_i(t) \big), \label{eq:dist_observer} 
\end{align}
where $\mathbf{K}_i$ is the gain matrix of the observer at actuation node $i$. Each actuation node $i$ for $i=1, 2, \hdots, q$ then applies the control input specified as
$u_{i}=\mathbf{f}_{i}^T\hat{\mathbf{x}}_i$ where $\mathbf{f}_{i} \in \mathbb{R}^n$ is the control gain for state-feedback at node $i$. 

We assume that the actuation nodes can either communicate their applied actuation signal $u_{i}(t)$ or their estimates $\hat{\mathbf{x}}_i$ among themselves, so that the term $\mathbf{u}(t)=[u_1(t)~ u_2(t)~\cdots u_q(t)]^T$ in \eqref{eq:dist_observer} is identical for each observer. This communication assumption is vital for the separation principle to hold in this distributed setting. Now, if the pair $(\mathbf{C}_{Oi},-\mathbf{L})$ is observable for all the actuation nodes $i=1, 2, \cdots, q$, then each actuation node can design its own observer gain matrix $\mathbf{K}_i$ using any standard pole-placement techniques that place the eigenvalues of $-(\mathbf{L}+\mathbf{K}_i\mathbf{C}_{Oi})$ sufficiently far left of the imaginary axis to make the estimation errors become zero quickly. Thereupon, the controller gain for all the actuation nodes can be obtained by finding the gain matrix $\mathbf{F}$ using the algorithm given in \cite{al2022observability}, where $\mathbf{F}=[\mathbf{f}_{1} ~\mathbf{f}_{2} \cdots \mathbf{f}_{q}]^T$. We formalize the outcome of this solution as follows.

\vspace{10pt}
\begin{theorem}
Consider the network synchronization model \eqref{eq:dynamics1}. Assume that: 1) the eigenvalues of $\mathbf{L}$ are distinct; 2) $q \geq m + 2$; 3) the pair $(\mathbf{C}_{Oi},-\mathbf{L})$ is observable for actuation nodes $i=1, 2, \cdots, q$; 4) the pair $(-\mathbf{L}, \mathbf{B})$ is controllable; and 5) all actuation nodes can communicate with each other.
Then, the state observer given by the gain matrix $\mathbf{K}_i$ and the state feedback gain $\mathbf{f}_i$ for all actuation nodes can be designed independently such that the estimation error converges to zero rapidly and observability of the closed-loop model is blocked at the compromised nodes. Specifically, the pair $(\mathbf{C}_A,\; -(\mathbf{L} + \mathbf{B}\mathbf{F}))$ has an unobservable mode at $(-\lambda_p)$ where $\mathbf{F}=[\mathbf{f}_{1} ~\mathbf{f}_{2} \cdots \mathbf{f}_{q}]^T$ and $\lambda_p$ is any selected eigenvalue of $\mathbf{L}$. Furthermore, all open-loop eigenvalues will be preserved. 
\end{theorem}

\begin{proof}
This proof follows an argument similar to that of Theorem 2. Since both the plant and all observers share the same input
$\mathbf{u}(t) = [u_1(t)~ u_2(t)~\cdots~ u_q(t)]^T$, the system matrix of the augmented dynamics for 
$[\mathbf{x}^T~ \mathbf{e}_1^T~\cdots~\mathbf{e}_q^T]^T$ can be written as
\begin{align}
\begin{bmatrix}
-(\mathbf{L}+\mathbf{B}\mathbf{F}) & * & \cdots & *\\
0 & -(\mathbf{L} + \mathbf{K}_1 \mathbf{C}_{O1}) & \cdots & 0 \\
\vdots & \vdots & \ddots & \vdots \\
0 & 0 & \cdots & -(\mathbf{L}+ \mathbf{K}_q \mathbf{C}_{Oq})
\end{bmatrix} \nonumber,
\end{align}
where $\mathbf{e}_i = \mathbf{x} - \hat{\mathbf{x}}_i$ for $i = 1,2,\dots,q$. From the block structure of the above matrix and by using the same argument as in Theorem 2, the result follows. 
\end{proof}

We remark that the gain matrix 
$\mathbf{F} = [\mathbf{f}_{1} ~ \mathbf{f}_{2} ~ \cdots ~ \mathbf{f}_{q}]^{T}$ 
obtained in this distributed approach is identical to the matrix 
$\mathbf{F}_2$ derived for the observer-based model under the same network setting and the same selection of eigenvalues and eigenvectors. Hence, both the centralized and distributed approaches yield the same controller gain matrix. We omit the presentation of the algorithm for computing this gain matrix here to avoid repetition.

Similar to the previous approaches, the number of required actuation nodes in Theorem 3 can be reduced to $q \geq |\mathcal{V}_{cut}| + 1$ when there exists a vertex cutset $\mathcal{V}_{cut}$ that separates the actuation nodes from the compromised nodes. We omit a formal statement of this result due to space constraints and to avoid repetition. Moreover, the number of required actuation nodes can be reduced by one if all eigenvalues of $\mathbf{L}$ are real.

Condition 5 above is not overly restrictive when the selected actuation nodes are neighbors. In such a case, they already communicate and share their state information. In addition, they only need to share their applied actuation signals $u_i(t)$ or their state estimates $\hat{\mathbf{x}}_i$. Otherwise, the separation principle does not hold, meaning that the observer and controller cannot be designed independently. 

We note that it is possible for the pair $(\mathbf{C}_{Oi},-\mathbf{L})$ to not be individually observable for all actuation nodes $i=1, 2, \cdots, q$, but it is collectively observable in the sense that  $([\mathbf{C}_{O1}^T~\mathbf{C}_{O2}^T~\cdots~\mathbf{C}_{Or}^T]^T,-\mathbf{L})$ is observable. In that case, each observer needs to update its equation as 
\begin{align}
\dot{\hat{\mathbf{x}}}_i(t)
    = -\mathbf{L}\hat{\mathbf{x}}_i(t) + \mathbf{B}\mathbf{u}(t)
       &+ \mathbf{K}_i\big( \mathbf{y}_{Oi}(t) - \mathbf{C}_{Oi}\hat{\mathbf{x}}_i(t)) \nonumber \\  &+ \sum_{j \in \mathcal{N}'_i} a_{ij}  (\hat{\mathbf{x}}_j - \hat{\mathbf{x}}_i), 
\end{align}
where $\mathcal{N}'_i$ denotes the set of all other actuation nodes with which node $i$ can communicate.  Given that this communication graph is a spanning tree on the network graph $\mathcal{G}$ and $a_{ij} > 0$, the consensus term 
$\sum_{j \in \mathcal{N}'_i} a_{ij} (\hat{\mathbf{x}}_j - \hat{\mathbf{x}}_i)$ 
drives the state estimates to converge to the true value, provided that the matrix $-(\mathbf{I}_n \otimes \mathbf{L} + \mathrm{diag}(\mathbf{K}_i \mathbf{C}_{oi}) + \mathcal{A} \otimes \mathbf{I})$
is Hurwitz \cite{han2018simple,wang2017distributed}. Here, $\otimes$ denotes the Kronecker product, and $\mathcal{A}$ denotes the Laplacian matrix associated with $\{a_{ij}\}$. We refer the reader to \cite{han2018simple,wang2017distributed} for further details. Unlike the case in Theorem 3, the observers must now be designed jointly, either heuristically or via linear matrix inequality (LMI) techniques. As the estimation error converges to zero, the closed-loop dynamics reduce to the state matrix $-(\mathbf{L} + \mathbf{B}\mathbf{F})$, where $\mathbf{F}$ can be designed in the same manner using the algorithm presented in \cite{al2022observability}. We leave the development of rigorous and provable algorithms for a fully distributed solution to the observability blocking problem in future work.

\section {Simulation Results}

In this section, we present numerical simulations that illustrate the effectiveness of the proposed solutions in blocking observability of a selected mode at specified compromised nodes. We consider a network model of 11 nodes, as shown in Fig. 1. The graph is undirected, and each edge is assigned a weight of $1$.

\begin{figure} [htbp]
    \centering
    \includegraphics[width=0.9\linewidth]{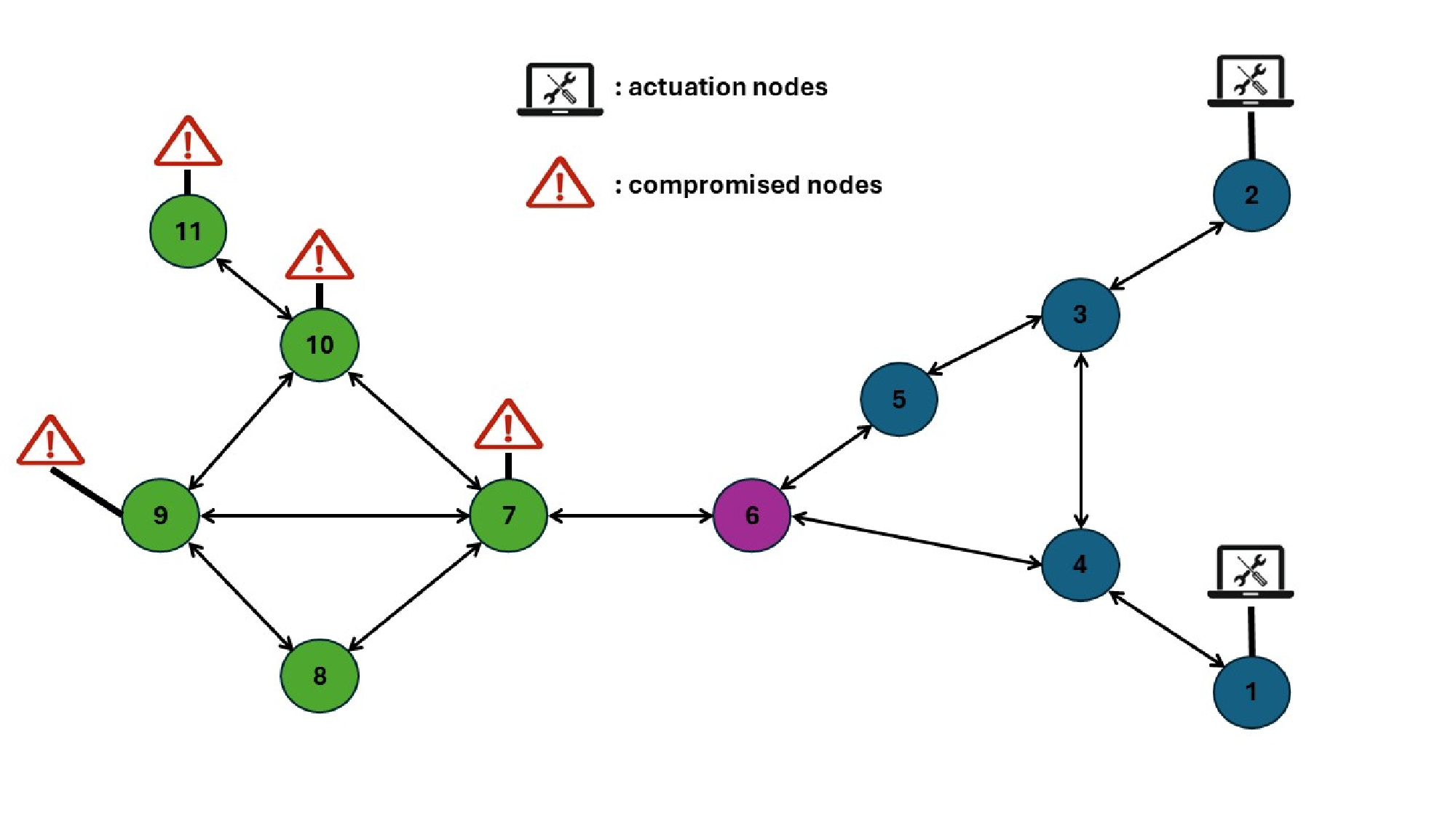}
    \caption[Model for Output Feedback Simulation]{11-node network}
    \label{fig: example network model}
\end{figure}

The actuation nodes are Nodes~$\{1, 2\}$ and the compromised nodes are Nodes~$\{7, 9, 10, 11\}$. Node 6 represents a single vertex-cutset that partitions the compromised and actuation nodes.
Therefore, if observability is blocked at Node 6, observability will also be blocked at the compromised nodes using two actuation nodes. Since all eigenvalues of $\mathbf{L}$ are real, $q = 2$ actuation nodes are sufficient to block observability at the compromised nodes, as implied by Corollaries 1 and 3. For our first simulation, we consider the case where the sensor nodes are located at Nodes $\{4, 5, 6, 7, 8\}$. We note that both $(\mathbf{C}_O, -\mathbf{L})$ and $(\mathbf{C}_A, -\mathbf{L})$ are observable, and $(-\mathbf{L}, \mathbf{B})$ is controllable.

Now, we implement Algorithm 1 to design observability-blocking controls based on output feedback. From the eigenvalues of $\mathbf{L}$, we specifically choose $\Lambda_D = [0;\ 0.1983;\ 0.6616;\ 0.8536;\ 1.2169]$ and $\lambda_{i_p} =0.1983$. The gain matrix $\mathbf{F}_1$ is obtained using Algorithm 1, which yields the controller gains at Nodes 1 and 2 as 
$[11.7216;$ $14.412;$ $-16.6171;$ $-26.6513;$ $17.1349]$ and
$[-9.7809;$ $-12.026;$ $13.8660;$ $22.2389;$ $-14.2980]$.
The corresponding eigenvector for $\lambda_{i_p}$ is now modified to $[0.6389;\ -0.7006;\ -0.25;\ 0.1388;\ -0.1388;\ 0;\ 0;\ 0;\ 0;$ $\ 0;\ 0]$ in the closed-loop system. Note that the entries of this closed-loop eigenvector corresponding to the compromised nodes are zero. Therefore, the closed-loop dynamics are now unobservable to the measurements at the compromised nodes. In fact, measurements taken at any nodes in $\mathcal{V}_{cut} \cup \mathcal{V}_2$ $=\{6, 7, 8, 9, 10, 11\}$ will find the mode $(-\lambda_{i_p})$ unobservable. Although the eigenvalues in $\Lambda_D$ are preserved, the rest of the eigenvalues are modified to different values. 

Now, consider the case where the sensor nodes are located at Nodes $\{6, 7, 8\}$. For this placement, the pair $(\mathbf{C}_O, -\mathbf{L})$ remains observable. However, Algorithm 1 now fails to compute $\mathbf{F}_1$ since the matrix $\mathbf{C}_O \mathbf{V}_h$ is not invertible. In fact, for this sensor placement, an output-feedback controller does not exist for blocking observability at the compromised nodes, as predicted by Corollary 2. However, we can still apply the observer-based method to block observability at the compromised nodes for this sensor placement. In this approach, we first build an observer using measurements obtained from the sensor nodes. Using any standard pole-placement technique (for example, MATLAB's \textit{place} function), we place the eigenvalues of the observer gain matrix $\mathbf{K}$ at $\{-1,-2, \cdots,-11\}$. The estimation error of the observer for each state is shown in Fig. 2, where the observer is initialized with zero initial conditions. As shown, the estimation error converges rapidly to zero. Although a faster convergence can be achieved by placing the eigenvalues further to the left, this would result in greater gain values for $\mathbf{K}$, which may lead to large control inputs $\mathbf{u}(t)$ due to increased transient overshoot in the estimation error. Therefore, the eigenvalues associated with $\mathbf{K}$ should be placed with caution. 

\begin{figure} [htbp]
    \centering
    \includegraphics[width=0.9\linewidth]{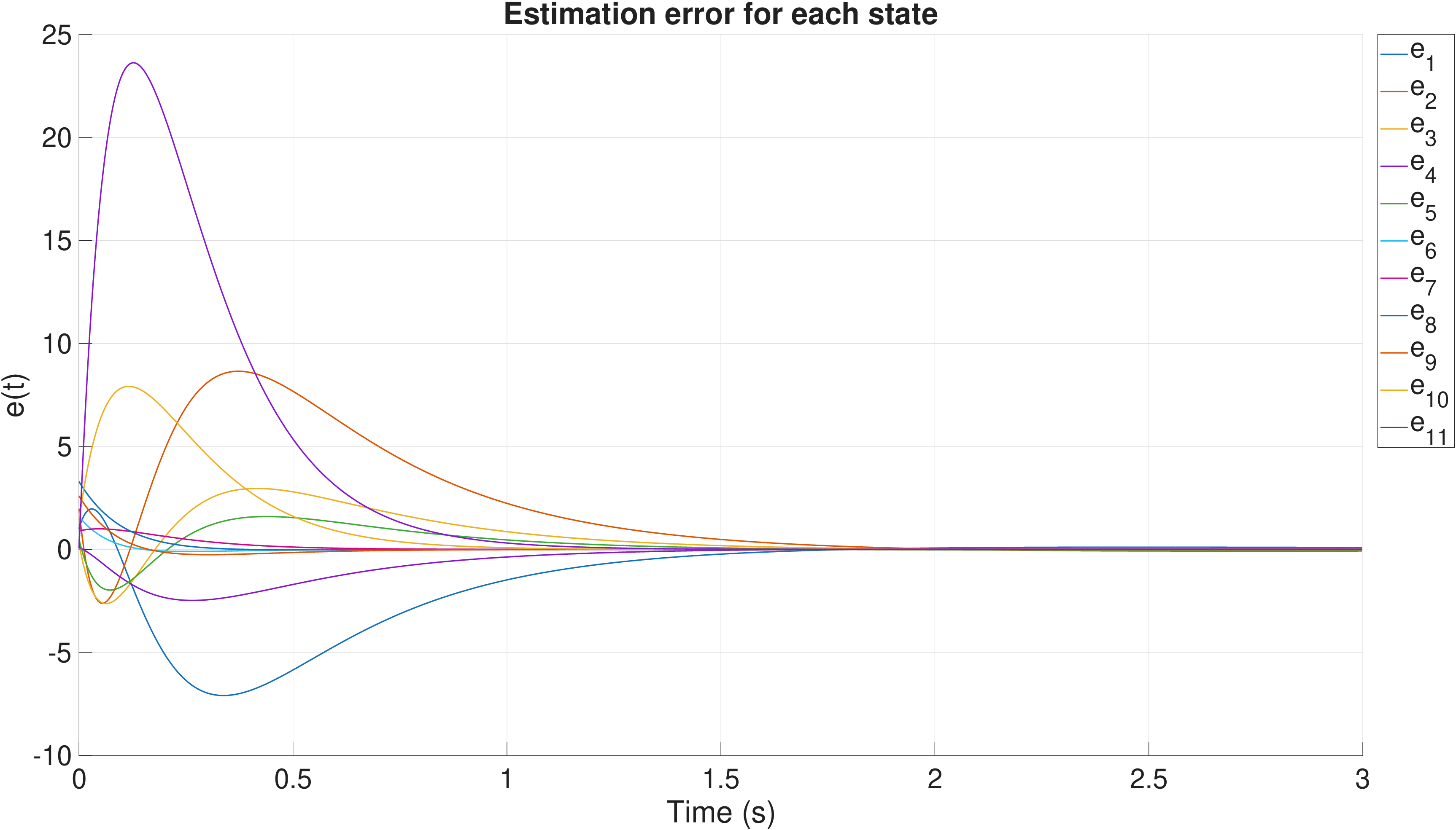}
    \caption{Evolution of estimation errors $\mathbf{e}(t) = \mathbf{x} - \hat{\mathbf{x}}(t)$ for each state}.
    \label{fig: example network model}
\end{figure}
We again choose $\lambda_p=0.1983$ and modify its corresponding eigenvector to $[0.6389;\ -0.7006;$ $\ -0.25;\ 0.1388;$ $\-0.1388;\ 0;\ 0;$ $\ 0;\ 0;\ 0;\ 0]$ to  impose unobservability. The gain matrix $\mathbf{F}_2$ is obtained using the algorithm in \cite{al2022observability}, which yields the controller gains at Nodes 1 and 2 as  $[0.8873;$ $1.0513;$ $ 0.8428;$ $0.7033;$ $0.6067;$ $0.2503;$ $-0.6088; -0.8012;$ $ -0.8347;$ $-0.9287;$ $-1.1584]$ and $[-0.7320;$ $-0.8773;$ $-0.7033;$ $-0.5869;$ $-0.5063; -0.2089;$ $0.5080;$ $0.6685;$ $0.6965;$ $0.7749;$ $0.9667]$. We note that all the open-loop eigenvalues and all the eigenvectors except the one corresponding to $\lambda_p=0.1983$ are preserved. 

We now implement the proposed distributed framework for observability blocking. Based on the graph in Fig. 1, we define the output matrix for actuation nodes 1 and 2 as $\mathbf{C}_{O1}=[\mathbf{e}_1~~ \mathbf{e}_4]^T$ and $\mathbf{C}_{O2}=[\mathbf{e}_2~~ \mathbf{e}_3]^T$, respectively. We note that both $(\mathbf{C}_{O1}, -\mathbf{L})$ and $(\mathbf{C}_{O2}, -\mathbf{L})$ are individually observable. Thus, Theorem 3 can be applied. We build two observers at actuation nodes 1 and 2. The observer estimation errors in these observers for each state are shown in Fig. 3, where the initial conditions for the observers are set to zero.

\begin{figure} [htbp]
    \centering
    \includegraphics[width=0.9\linewidth]{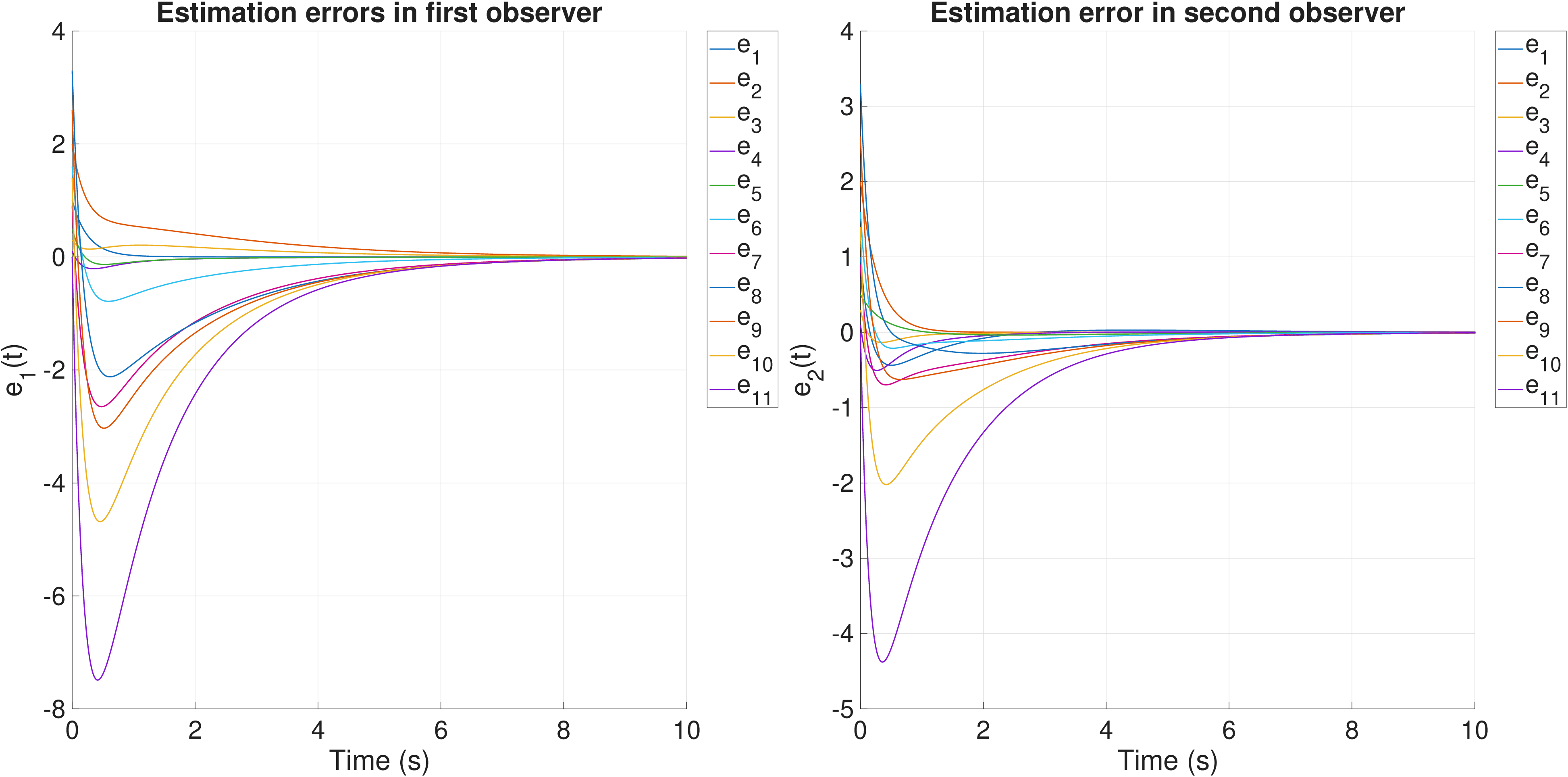}
    \caption{Evolution of estimation errors $\mathbf{e}(t) = \mathbf{x} - \hat{\mathbf{x}}(t)$ in observers at Node 1 and Node 2 for each state.}
    \label{fig: example network model2}
\end{figure}

We again choose $\lambda_p=0.1983$ and modify its corresponding eigenvector to the same vector $[0.6389;\ -0.7006;\ -0.25;\ 0.1388;\ -0.1388;\ 0;\ 0;\ 0;\ 0;\ 0]$ to impose unobservability. The controller gains at Nodes 1 and 2 are identical to the previous case with one centralized observer, as remarked in Section III-C.  

\section{Conclusion}

In this paper, we developed output-feedback and observer-based control strategies to block observability in networked systems using partial state information. The observer-based approach preserves all open-loop eigenvalues and offers greater flexibility, while the output-feedback approach provides a simpler alternative with partial eigenstructure preservation. We also extended the observer-based design to a distributed framework. 



	\bibliographystyle{ieeetr}
	\bibliography{references}

\begin{thebibliography}{10}

\bibitem{rahmani2009controllability}
A.~Rahmani, M.~Ji, M.~Mesbahi, and M.~Egerstedt, ``Controllability of multi-agent systems from a graph-theoretic perspective,'' {\em SIAM Journal on Control and Optimization}, vol.~48, no.~1, pp.~162--186, 2009.

\bibitem{pasqualetti2014controllability}
F.~Pasqualetti, S.~Zampieri, and F.~Bullo, ``Controllability metrics, limitations and algorithms for complex networks,'' {\em IEEE Transactions on Control of Network Systems}, vol.~1, no.~1, pp.~40--52, 2014.

\bibitem{summers2015submodularity}
T.~H. Summers, F.~L. Cortesi, and J.~Lygeros, ``On submodularity and controllability in complex dynamical networks,'' {\em IEEE Transactions on Control of Network Systems}, vol.~3, no.~1, pp.~91--101, 2015.

\bibitem{roy2016sensor}
S.~Roy, J.~A. Torres, and M.~Xue, ``Sensor and actuator placement for zero-shaping in dynamical networks,'' in {\em 2016 IEEE 55th Conference on Decision and Control (CDC)}, pp.~1745--1750, IEEE, 2016.

\bibitem{li2020structural}
J.~Li, X.~Chen, S.~Pequito, G.~J. Pappas, and V.~M. Preciado, ``On the structural target controllability of undirected networks,'' {\em IEEE Transactions on Automatic Control}, vol.~66, no.~10, pp.~4836--4843, 2020.

\bibitem{paridari2017framework}
K.~Paridari, N.~O’Mahony, A.~E.-D. Mady, R.~Chabukswar, M.~Boubekeur, and H.~Sandberg, ``A framework for attack-resilient industrial control systems: Attack detection and controller reconfiguration,'' {\em Proceedings of the IEEE}, vol.~106, no.~1, pp.~113--128, 2017.

\bibitem{sridhar2011cyber}
S.~Sridhar, A.~Hahn, and M.~Govindarasu, ``Cyber--physical system security for the electric power grid,'' {\em Proceedings of the IEEE}, vol.~100, no.~1, pp.~210--224, 2011.

\bibitem{xue2014security}
M.~Xue, W.~Wang, and S.~Roy, ``Security concepts for the dynamics of autonomous vehicle networks,'' {\em Automatica}, vol.~50, no.~3, pp.~852--857, 2014.

\bibitem{zhang2023observability}
Y.~Zhang, R.~Cheng, and Y.~Xia, ``Observability blocking for functional privacy of linear dynamic networks,'' in {\em 2023 62nd IEEE Conference on Decision and Control (CDC)}, pp.~7469--7474, IEEE, 2023.

\bibitem{zhang2022privacy}
J.~Zhang, J.~Lu, and X.~Chen, ``Privacy-preserving average consensus via edge decomposition,'' {\em IEEE Control Systems Letters}, vol.~6, pp.~2503--2508, 2022.

\bibitem{al2022observability}
A.~Al~Maruf and S.~Roy, ``Observability-blocking control using sparser and regional feedback for network synchronization processes,'' {\em Automatica}, vol.~146, p.~110586, 2022.

\bibitem{tran2025observability}
J.~D. Tran and A.~Al~Maruf, ``Observability-blocking controls for double-integrator and higher order integrator networks,'' in {\em 2025 American Control Conference (ACC)}, pp.~3279--3285, IEEE, 2025.

\bibitem{anguluri2025mode}
R.~Anguluri and A.~Al~Maruf, ``Mode participation and inter-area-observability blocking controllers for power networks,'' in {\em 2025 IEEE Conference on Control Technology and Applications (CCTA)}, pp.~762--767, IEEE, 2025.

\bibitem{al2019observability}
A.~Al~Maruf and S.~Roy, ``Observability-blocking controllers for network synchronization processes,'' in {\em 2019 American Control Conference (ACC)}, pp.~2066--2071, IEEE, 2019.

\bibitem{zhang2017distributed}
D.~Zhang, S.~K. Nguang, and L.~Yu, ``Distributed control of large-scale networked control systems with communication constraints and topology switching,'' {\em IEEE Transactions on Systems, Man, and Cybernetics: Systems}, vol.~47, no.~7, pp.~1746--1757, 2017.

\bibitem{srinathkumar2003eigenvalue}
S.~Srinathkumar, ``Eigenvalue/eigenvector assignment using output feedback,'' {\em IEEE Transactions on Automatic Control}, vol.~23, no.~1, pp.~79--81, 2003.

\bibitem{andry2007eigenstructure}
A.~N. Andry, E.~Y. Shapiro, and J.~Chung, ``Eigenstructure assignment for linear systems,'' {\em IEEE transactions on aerospace and electronic systems}, no.~5, pp.~711--729, 2007.

\bibitem{rugh1996linear}
W.~J. Rugh, {\em Linear system theory}.
\newblock Prentice-Hall, Inc., 1996.

\bibitem{siljak2011decentralized}
D.~D. Siljak, {\em Decentralized control of complex systems}.
\newblock Courier Corporation, 2011.

\bibitem{kailath2000linear}
T.~Kailath, A.~H. Sayed, and B.~Hassibi, {\em Linear estimation}.
\newblock Prentice Hall, 2000.

\bibitem{chen1984linear}
C.-T. Chen, {\em Linear system theory and design}.
\newblock Saunders college publishing, 1984.

\bibitem{han2018simple}
W.~Han, H.~L. Trentelman, Z.~Wang, and Y.~Shen, ``A simple approach to distributed observer design for linear systems,'' {\em IEEE Transactions on Automatic Control}, vol.~64, no.~1, pp.~329--336, 2018.

\bibitem{wang2017distributed}
L.~Wang and A.~S. Morse, ``A distributed observer for a time-invariant linear system,'' {\em IEEE Transactions on Automatic Control}, vol.~63, no.~7, pp.~2123--2130, 2017.

\end{thebibliography}

\end{document}